\documentclass[conference]{IEEEtran}
\usepackage{bm,cite,algorithm,algorithmic,float,amsmath,amssymb}

\usepackage{amssymb}
\usepackage{amsmath}
\usepackage{graphicx}
\usepackage{cite}
\usepackage{citesort}
\usepackage{balance}
\usepackage[utf8x]{inputenc}
\usepackage{amsthm}

\newtheorem{theorem}{Theorem}
\newtheorem{lemma}{Lemma}

\IEEEoverridecommandlockouts

\usepackage{graphicx,epstopdf}
\usepackage{epsfig}	
\usepackage{amsfonts,balance}
\usepackage{bbm}
\floatname{algorithm}{Algorithm}
\begin{document}

\title{Variance-Constrained Capacity of the Molecular Timing Channel with Synchronization Error}

\author{\authorblockN{Malcolm Egan$^{\dag}$, Yansha Deng$^\ddag$, Maged Elkashlan$^{\ddag}$, and Trung Q. Duong$^{\ast}$
\authorblockA{\authorrefmark{2}\footnotesize Faculty of Electrical Engineering, Czech Technical University in Prague, Czech Republic }
\authorblockA{\authorrefmark{3}\footnotesize School of Electronic Engineering and Computer Science, Queen Mary University of London, UK }
\authorblockA{\authorrefmark{1}\footnotesize School of Electronic Engineering and Computer Science, Queen's University Belfast, UK }
}}

\maketitle

\begin{abstract}
Molecular communication is set to play an important role in the \textit{design} of complex biological and chemical systems. An important class of molecular communication systems is based on the timing channel, where information is encoded in the delay of the transmitted molecule---a synchronous approach. At present, a widely used modeling assumption is the perfect synchronization between the transmitter and the receiver. Unfortunately, this assumption is unlikely to hold in most practical molecular systems. To remedy this, we introduce a clock into the model---leading to the molecular timing channel with synchronization error. To quantify the behavior of this new system, we derive upper and lower bounds on the variance-constrained capacity, which we view as the step between the mean-delay and the peak-delay constrained capacity. By numerically evaluating our bounds, we obtain a key practical insight: the drift velocity of the clock links does not need to be significantly larger than the drift velocity of the information link, in order to achieve the variance-constrained capacity with perfect synchronization.
\end{abstract}

\maketitle

\section{Introduction}

With the rise of synthetic biology and chemistry, new applications are abundant: vaccines for malaria; biofuels; and even manipulation of bacteria colony populations \cite{You2004}. Despite early successes, an improved understanding of the underlying mechanisms of complex biological networks is required to go further. Communication is a fundamental feature of many of these mechanisms: biological networks rely heavily on communication between different components. In contrast with traditional cellular wireless systems, communication is often between nano-scale devices with unreliable energy sources. This means that information is carried by molecules with messages encoded in concentration levels or transmission delays.

An important class of molecular communication systems are those based on diffusion, where molecules carrying information propagate via the random motion induced by the collisions with the fluid molecules. A key example is in pheromonal communication \cite{Bossert1963}. Recently, a variety of diffusion-based communication mechanisms have been proposed (see e.g., \cite{Pierobon2013,Kadloor2012}). Within these mechanisms, the molecular timing channel has the potential to offer the highest transmission rates, assuming that the channel is sufficiently reliable. This is achieved by encoding information in the delay between the time the molecule is transmitted and the beginning of the time slot \cite{Kadloor2012,eckford2007nanoscale}. An important model for the molecule timing channel is the additive inverse Gaussian noise (AIGN) channel, which has been evaluated via capacity bounds with constraints on the both the average delay \cite{Srinivas2012} and the peak-delay \cite{Chang2012}.

A key assumption in molecular timing channel models is that the transmitter and receiver agree on when the information molecule is sent; i.e., communication is synchronous. Unfortunately, this assumption does not hold in general. In contrast with other communication networks---such as wireless cellular networks---it is not easy to share global clock information throughout the  molecular network.

In this paper, we introduce the molecular timing channel model with synchronization error. The basis of our new model is a global clock that sends molecules to synchronize the system, which is added to the standard timing channel. Our model consists of three molecular links: the transmitter-receiver timing channel; the clock-transmitter link; and the clock-receiver link. The transmitter and receiver are informed when a time slot begins based on the arrival of the clock molecule. Importantly, the clock molecule will typically not arrive at the same time, which leads to synchronization error.

In contrast to the standard AIGN timing channel model in \cite{Srinivas2012,Chang2012}, the timing channel with synchronization error must  accommodate the information molecule arriving before the start of the receiver's time slot. As such, standard bounds on the capacity are not directly applicable. To overcome this problem, we evaluate the molecular timing channel with synchronization error in terms of the variance-constrained capacity; i.e., the capacity when the distribution of the delay (corresponding to messages) has both mean and variance constraints. This is achieved via new upper and lower variance-constrained capacity bounds. Importantly, our approach is closely related to the peak-delay constrained capacity, and as such can be used to guide system design.

Numerical evaluation of our bounds suggests that the variance-constrained capacity can be highly dependent on the information (transmitter to receiver) link as well as the clock-transmitter link and the clock-receiver link. Fortunately, our numerical results also suggest that in order to achieve the capacity upper bound with perfect synchronization, the drift velocities of the clock links do not need to be significantly larger than the drift velocity of the information link.

\section{System Model}

\subsection{Network Topology and Synchronization}

Consider the molecular timing channel illustrated in Fig.~\ref{fig:sys_model_illus}, consisting of a global clock, a transmitter, and a receiver. In each time slot, the transmitter sends a single information molecule to the receiver. The message is encoded in the time between the onset of the current time slot and when the information molecule is emitted. The information molecule then diffuses through the fluid medium, eventually arriving at the receiver.

\begin{figure}[h!]
   \begin{center}
        \includegraphics[width=3in]{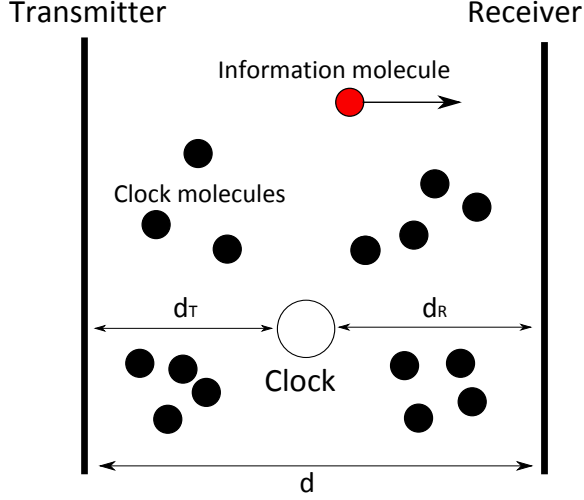}
        \caption{The synchronous molecular timing channel consisting of a global clock, a transmitter, and a receiver.}
        \label{fig:sys_model_illus}
    \end{center}
\end{figure}

Typically, molecular communication systems have severe complexity constraints. The lack of processing capability means that the molecular transmitter and receiver cannot independently determine the beginning of the time slot. As such, the beginning of a time slot is signaled by the global clock via the emission of synchronization molecules, which do not interact with the information molecule and can be detected by the transmitter and receiver.

Our model is based on the following assumptions.
\begin{enumerate}
\item The physical network consisting of the transmitter, the receiver, and the global clock is two-dimensional. We denote the distance between the transmitter and the receiver as $d$, the distance between the clock and the transmitter as $d_T$, and the distance between the clock and the receiver as $d_R$.
\item The transmitter and the receiver forms parallel lines, which partition $\mathbb{R}^2$.
\item Any information molecule or synchronization molecule that hits the receiver boundary is permanently absorbed.
\item There is zero friction in the fluid media.
\item The synchronization molecules are emitted from the global clock at time $t = 0$; i.e., the synchronization clock is reset at the onset of a time slot.
\item The transmitter can perfectly control the release time of each information molecule.
\item The receiver can perfectly distinguish between the information and clock molecules.
\item There is no interference from the molecules transmitted in previous time slots.
\end{enumerate}

The assumption that the receiver is partitioned by a line implies that the time taken for the molecule to diffuse from the transmitter to the receiver is only dependent on how fast the molecule travels along the line perpendicular to the receiver. As such, the position process that determines the hitting time is one-dimensional.

Our assumption of no interference from molecules transmitted in previous time slots is practical in many molecular channels. This is due to the fact that power constraints (arising from unreliable energy-harvesting mechanisms) can only support sparse transmissions. A similar situation also occurs in wireless nano-sensor networks \cite{Sung2014}.

In light of the zero friction assumption, a good model for the position of the information and clock molecules is the Wiener process, $W(x)$ \cite{Srinivas2012}. In particular, $W(x)$ is the continuous-time process with position increments $R_i = W(x_{i-1}) - W(x_i)$, where $0 \leq x_1 < x_2 < \cdots < x_k$ is a sequence of time instants arranged in increasing order. Moreover, $R_i \sim N(v(x_i - x_{i-1}), \sigma^2(x_i - x_{i-1}))$, where $v$ is the drift\footnote{Drift is typically caused by a difference in concentration between the transmitter and the receiver.} velocity of the fluid medium and the variance $\sigma^2 = D/2$ is governed by the diffusion coefficient\footnote{$D$ typically lies between $1 - 10 \mu {m^2}$ \cite{Berthier2010}.} $D$.

An important consequence of modeling the diffusion of the synchronization molecules as random processes is that there will be a delay between emission and reception. As such, the transmitter and receiver will not agree on the precise starting time of the current time slot, which introduces \textit{synchronization error}.

\subsection{Channel Model, Encoding and Decoding}

At the beginning of the time slot (from the perspective of the global clock), clock molecule is transmitted throughout the molecular network. We assume that the channels between the clock and both the transmitter and the receiver are governed by the Wiener process. As such, there will be a random delay of $E_T$ and $E_R$, before the clock molecule arrives at the transmitter and receiver, respectively. In particular, $E_T$ and $E_R$ are distributed according to the inverse Gaussian (IG) distribution\footnote{The inverse Gaussian  distribution corresponds to the distribution of the first hitting time of the Wiener process \cite{Srinivas2012}.}, with probability density functions (pdfs)
\begin{align}
f_{E_T}(x) =\left\{
\begin{array}{ll}
\sqrt {\frac{\lambda_T }{{2\pi {x^3}}}} \exp \left( { - \frac{{\lambda_T {{\left( {x - \mu_T} \right)}^2}}}{{2{\mu_T ^2}x}}} \right)
 &\mbox{$x > 0$} ;\\
0 &\mbox{$x < 0$},\;
\end{array}
\right.
\end{align}
and
\begin{align}
f_{E_R}(x) =\left\{
\begin{array}{ll}
\sqrt {\frac{\lambda_R }{{2\pi {x^3}}}} \exp \left( { - \frac{{\lambda_R {{\left( {x - \mu_R} \right)}^2}}}{{2{\mu_R ^2}x}}} \right)
 &\mbox{$x > 0$} ;\\
0 &\mbox{$x < 0$},\;
\end{array}
\right.
\end{align}
where $\mu_T = \frac{d_T}{v_C}$ and $\lambda_T = \frac{d_T^2}{\sigma^2} $  are parameters of the pdf of $E_T$, and $\mu_R = \frac{d_R}{v_C}$ and $\lambda_R = \frac{d_R}{\sigma^2}$  are parameters of the pdf of $E_R$.

Once the transmitter receives a clock molecule, the transmitter's time slot begins and a new message is transmitted. Let $\{1,2,\ldots,M\}$ be the finite set of possible messages. To encode a message $m \in \{1,2,\ldots,M\}$, the transmitter maps the message to a delay $X$. The time that the transmitter emits the information molecule (according to the global clock) is then $X + E_T$.

Next, the information molecule diffuses through the channel between the transmitter and the receiver (the information link), and travels for a random time $N$, which is distributed according to the IG distribution; similar to  the clock molecule channels. The pdf of $N$ is then
\begin{align}
f_{N}(x) =\left\{
\begin{array}{ll}
\sqrt {\frac{\lambda_N }{{2\pi {x^3}}}} \exp \left( { - \frac{{\lambda_N {{\left( {x - \mu_N} \right)}^2}}}{{2{\mu_N ^2}x}}} \right)
 &\mbox{$x > 0$} ;\\
0 &\mbox{$x < 0$},\;
\end{array}
\right.
\end{align}
where $\mu_N = \frac{d}{v_I}$ and$ ~\lambda_N = \frac{d^2}{\sigma^2}$.

After traversing the channel, the information molecule is absorbed by the receiver. From the perspective of the global clock, the absorption occurs at time $X + E_T + N$. However, the receiver observes delay $Y = X + E_T + N - E_R$ as the receiver's time slot only begins at time $E_R$ (according to the global clock). Observe that $Y \in (-\infty,\infty)$; that is, the information molecule can potentially arrive before the receiver identifies that a new time slot has begun. To cope with the situation where the information molecule arrives before the clock molecule, the receiver decodes
\begin{align}
Z = \max\{X + E_T + N - E_R,0\},
\end{align}
which ensures that the received signal lies in $[0,\infty)$.

\section{Variance-Constrained Capacity}\label{sec:capacity}

In this section, we derive new bounds on the capacity with mean and variance constraints of the molecular timing channel with synchronization error. In particular, we obtain both upper and lower bounds. Our lower bound is based on the data processing inequality, for which we provide an intuitive interpretation in terms of the signal processing capabilities of the receiver.

The variance-constrained capacity is defined as
\begin{align}\label{eq:var_cons_cap}
C = \max_{p_X(x):\mathbb{E}[X] = m,~m^2 \leq \mathbb{E}[X^2] \leq a} I(X;Z),
\end{align}
where $p_X(x)$ is a pdf and $Z = \max\{X + E_T + N - E_R,0\}$. We note that the variance constrained capacity is in fact related to the peak constrained capacity. To see this, observe that
\begin{align}\label{eq:holder}
\frac{\mathbb{E}[X^2]}{\mathbb{E}[X]} \leq X_{peak},
\end{align}
which follows from H{\"o}lder's inequality (a generalization of the Cauchy-Schwarz inequality) \cite{Hardy1952}. As such, it is possible to view the variance-constrained capacity as the step between the mean-delay constrained capacity and the peak-delay constrained capacity. This is important since capacity bounds with the peak-delay constrained input are difficult to obtain and are more complicated than the capacity bounds with the mean or variance constraints on the delay.

Although bounds are already well-established for the standard molecular timing channel, the synchronization error introduces new difficulties. In particular, as the receiver decodes based on $Z = \max\{X + E_T + N - E_R,0\}$ (instead of $X + N$ as in \cite{Srinivas2012,Chang2012}), the standard bounds cannot be directly applied. We solve this problem next.

\subsection{Capacity Upper Bound}

Our first key result is a new upper bound on the variance-constrained capacity (see (\ref{eq:var_cons_cap})). The statistics of $N,E_T$ and $E_R$ play key roles, as well as the mean and variance constraints $m$ and $a$, respectively.
\begin{theorem}[Capacity upper bound]\label{thrm:cap_ub}
An upper bound on the variance-constrained capacity (in nats) of the timing channel with synchronization error is given by
\begin{align}\label{upper bound}
C_{UB} &= \log ((m + \mu _N + \mu _{E_T})) + 1 - \min(g(c^*),0)\notag\\
&~~~ - h(N + E_T)\frac{m^2}{a} \int_0^m f_{E_R}(u) \left(1 - \frac{u}{m}\right)^2 du,
\end{align}
where
\begin{align}
g(c^*) = \min\left(h(N)\mathbf{1}_{h(N)<0},-\int_{c^*}^\infty f_N(u)\log f_N(u)du\right)
\end{align},
with
\begin{align}\label{eq:cstar}
c^* = \left\{
        \begin{array}{ll}
          0, & \textrm{if } f_N(u) < 1, \forall u; \\
          \inf\{u:f_N(u) = 1\}, & \textrm{else.}
        \end{array}
      \right.
\end{align}
\end{theorem}
\begin{proof}
See Appendix~\ref{app:thrm_cap_ub}.
\end{proof}

Although the expression in Theorem~\ref{thrm:cap_ub} is complicated, we can easily observe that both mean and variance constraints play an important role in the capacity upper bound. In particular, we observe that the upper bound increases with increasing the variance constraint. On the other hand, the mean constraint plays a more complicated role, which is examined via numerical evaluation in Section~\ref{sec:numerical}.

The numerical evaluation of Theorem~\ref{thrm:cap_ub} is simplified in the case that $\frac{\lambda_N}{\mu_N^2} = \frac{\lambda_{T}}{\mu_T^2}$. This is due to the additivity property of the IG distribution in \cite{chhikara1988inverse}, which guarantees that $N + E_T$ is IG distributed. The entropy of $N + E_T$ can then be easily computed via the following lemma \cite{Chang2012}. In general, the entropy $h(N + E_T)$ can be computed numerically.

\begin{lemma}[Entropy of inverse Gaussian distribution]\label{lem:entropy_IG}
The entropy of an inverse Gaussian distributed random variable with parameters $\mu$ and $\lambda$ is given by
\begin{align}
h_{IG(\mu,\lambda)} = \frac{1}{2}\log\frac{2\pi \mu_U^3}{\lambda} + \frac{3}{2}\exp\left(\frac{2\lambda}{\mu}\right)Ei\left(-\frac{2\lambda}{\mu}\right) + \frac{1}{2},
\end{align}
where
\begin{align}
Ei(-x) = -\int_x^\infty \frac{e^{-t}}{t}dt.
\end{align}
\end{lemma}

Further discussion of the behavior of our upper bound is provided in Section~\ref{sec:numerical}.

\subsection{Capacity Lower Bound}

We now turn our attention to the lower bound on the variance-constrained capacity. In general, a lower bound can be found simply by choosing a distribution $p_X(x)$ such that the mean and the variance constraints are satisfied. Unfortunately, accounting for synchronization error complicates concrete calculations by introducing a multi-dimensional integral. At present, the integral can only be evaluated numerically, even with $p_X(x)$ chosen to be exponentially or IG distributed. Although this numerical approach can provide important insights, it is highly desirable to obtain simpler bounds to obtain physical insights into the behavior of the system.

A key approach to obtaining tractable lower bounds on the capacity is via the data processing inequality. In essence, the receiver decodes a processed version of the observed delay $Z = \max\{X + E_T + N - E_R,0\}$.

An important special case of the data processing approach is when the receiver observes the random variable $B$ satisfying
\begin{align}
B = \left\{
      \begin{array}{ll}
        0, & \mathrm{if}~Z = 0 \\
        1, & \mathrm{if}~Z > 0.
      \end{array}
    \right.
\end{align}
Intuitively, this scenario corresponds to a receiver that decodes based on whether or not a molecule is absorbed within the time slot. In molecular communication systems with low computational capabilities or strict power constraints (due to unreliable energy harvesting), this is a practical solution.

To derive an expression for the variance-constrained capacity when the receiver decodes according to $B$. Observe that $X \leftrightarrow Z \leftrightarrow B$ forms a Markov chain. As such, the data processing inequality can be applied to yield
\begin{align}
I(X;Z) \geq I(X;B).
\end{align}
The lower bound on the variance-constrained capacity can then be obtained via $I(X;B) = h(B) - h(B|X)$, where
\begin{align}\label{eq:Cap_LB}
h(B) &= -F_Y(0)\log F_Y(0) - (1 - F_Y(0))\log (1 - F_Y(0)),\notag\\
h(B|X) &= -\int_0^\infty p_X(x)\left[ F_{Y|X = x}(0)\log F_{Y|X = x}(0)\right.\notag\\
&~~~\left.+ (1 - F_{Y|X = x}(0))\log (1 - F_{Y|X = x}(0))\right]dx,
\end{align}
where $Y = X + N + E_T - E_R$. The behavior of the bound is discussed in the next section.

\section{Numerical Results}~\label{sec:numerical}

In this section, we demonstrate the behavior of the molecular timing channel with synchronization error via numerical results. We show that the design of the clock has a significant effect on the performance of the system.

\begin{figure}[h!]
   \begin{center}
        \includegraphics[width=3.7in]{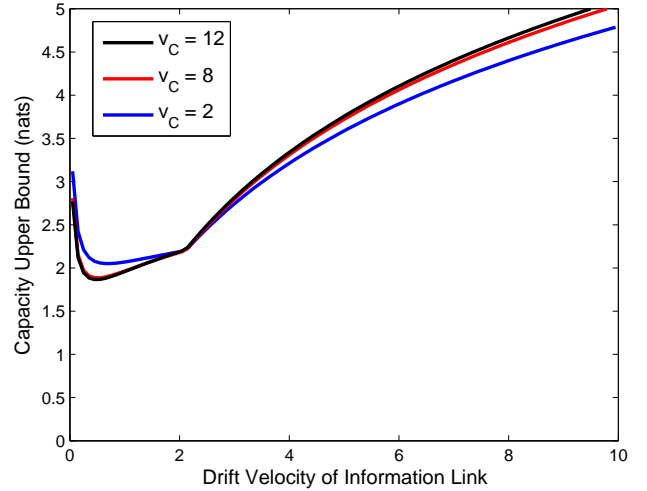}
        \caption{Plot of variance-constrained capacity upper bound versus the drift velocity of the information link, $v_I$, (between the transmitter and receiver). The three curves correspond to different drift velocities for the clock links (where the source is the clock). The mean constraint is $m = 3$ and the variance constraint is $a = 18$.}
        \label{fig:C_UBvsm}
    \end{center}
\end{figure}

Fig.~\ref{fig:C_UBvsm} plots the variance-constrained capacity upper bound versus the drift velocity of the information link $v_I$ for different velocities of the clock links $v_C$ (where the source is the clock), and $d = d_T = d_R = 1$ with $\sigma^2 = 1$. Observe that there  are diminishing gains as $v_C$  increases. Moreover, most of the capacity gains are made for relatively low clock-link drift velocities compared with the drift velocity of the information link. This is important as it means that the capacity with perfect synchronization can be achieved by relatively low clock link drift velocities.

Fig.~\ref{fig:C_UBvsm} also shows that the capacity bound is large at near-zero drift velocities for the information link. As such, the upper bound is not tight at low velocities of the information link. This is also consistent with the bounds obtained in \cite{Srinivas2012,Chang2012}, suggesting that an alternative approach is required to gain insights at low velocities.

\begin{figure}[h!]
   \begin{center}
        \includegraphics[width=3.2in]{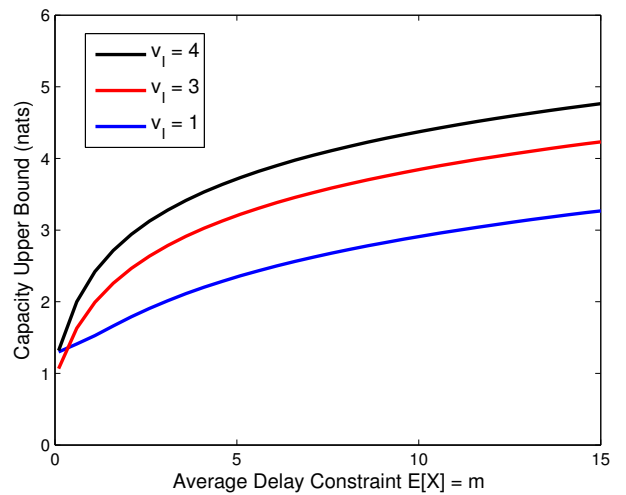}
        \caption{Plot of variance-constrained capacity upper bound versus the average delay constraint $\mathbb{E}[X] = m$ and variance constraint $\mathbb{E}[X^2] \leq 2m^2$, with $v_C = 4$.}
        \label{fig:C_UBvsdr}
    \end{center}
\end{figure}

Fig.~\ref{fig:C_UBvsdr} plots the variance-constrained capacity upper bound in \eqref{upper bound} as the average delay constraint $\mathbb{E}[X] = m$ varies. We consider  symmetric network with $d = d_T = d_R = 1$, $\sigma^2 = 1$,  varying drift velocity $v_I$ of the information link, and a drift velocity $v_C = 4$ for the clock links. Observe that significant gains can be achieved by increasing the drift velocity of the information link from $v_I = 1$ to $v_I = 3$. This suggests that the delay constraint needs to be carefully matched with the drift velocity in order to ensure efficient operation of the network.

\begin{figure}[h!]
   \begin{center}
        \includegraphics[width=3.7in]{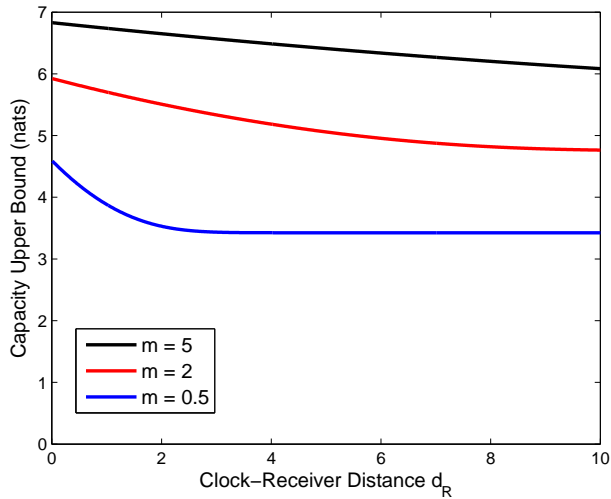}
        \caption{Plot of variance-constrained capacity upper bound versus the distance between the clock and the receiver $d_R$.}
        \label{fig:C_UBdist}
    \end{center}
\end{figure}

Fig~\ref{fig:C_UBdist} plots the capacity upper bound in \eqref{upper bound} as the distance between the clock and the receiver varies. The clock-transmitter distance $d_T = 0.1$, the transmitter-receiver distance $d = 0.1$, $\sigma^2 = 1$, and $a = 2m^2$. Observe that for the strictest constraint ($m = 0.5$) the capacity decreases rapidly as the distance $d_R$ increases, while for larger $m$ the capacity decreases slowly. This suggests that asymmetries in the network topology need to be carefully considered for small delay constraints.

\begin{figure}[h!]
   \begin{center}
        \includegraphics[width=3.2in]{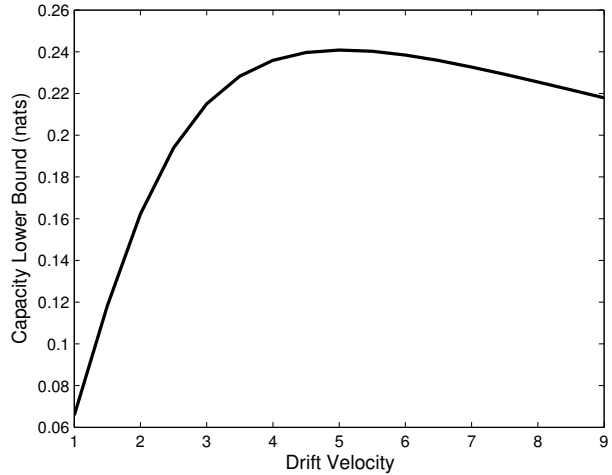}
        \caption{Plot of variance-constrained capacity lower bound  with varying drift velocity (common for all links). We set $p_X(x)$ to be exponentially distributed with $\mathbb{E}[X] = m$, $d = d_R = d_T = 1$, and $\sigma^2 = 1$.}
        \label{fig:C_LBvsv}
    \end{center}
\end{figure}

Fig.~\ref{fig:C_LBvsv} plots the capacity lower bound in (\ref{eq:Cap_LB}) as the drift velocity varies. We assume that the drift velocity is the same for each link and that the input distribution $p_X(x)$ is exponentially distributed with mean $m$. Observe that the lower bound is maximized for the drift velocity $v \approx 5$. At this velocity, approximately $1$ bit ($0.3$ nats) is achievable, which is the maximum possible as the encoding is binary. This explains why the capacity lower bound is significantly less than the capacity upper bound, which is not restricted to binary encoding. As such, there is a high cost that comes with low processing capabilities (potentially due to energy constraints).

\section{Conclusion}

Synchronization error is an important and previously neglected aspect of the molecular timing channel. We introduced a new model of the timing channel, which includes synchronization error induced via the random delay of molecules sent from a global clock. We then evaluated the model by deriving new upper and lower bounds on the variance-constrained capacity. We obtained important practical insights via numerical evaluation of our bounds. In particular, our results suggest that in order to achieve the capacity with perfect synchronization, the drift velocity of the clock links does not need to be significantly larger than the drift velocity of the information link. We also observed that the upper bound is not tight at low velocity, as in \cite{Srinivas2012,Chang2012}. Motivated by this observation, we are currently developing accurate low velocity capacity approximations.

\appendices

\section{Proof of Theorem~\ref{thrm:cap_ub}}\label{app:thrm_cap_ub}

\begin{proof}
Let $Z = \max\{X + N + E_T - E_R,0\}$ and $Z' = X+ E_T + N$. Clearly, $\mathbb{E}[Z] \leq \mathbb{E}[Z']$. Now denote, $R$ and $R'$ as exponentially distributed random variables with mean $\mathbb{E}[Z]$ and $\mathbb{E}[Z']$, respectively. Then
\begin{align}
I(X;Z) &= h(Z) - h(Z|X)\notag\\
&\overset{(a)}{\leq} h(R) - h(Z|X)\notag\\
&\overset{(b)}{\leq} h(R') - h(Z|X)\notag\\
&\overset{(c)}{\leq} h(R') - h(Z|X,E_R),
\end{align}
where $(a)$ follows from the fact that the exponential distribution maximizes entropy over all random variables with positive support and mean $\mathbb{E}[Z]$ \cite{Cover2006}, and $(b)$ follows from the fact that $\mathbb{E}[Z'] \geq \mathbb{E}[Z]$, the entropy of exponentially distributed random variables increases with the mean. $(c)$ follows since conditioning reduces entropy.

Now observe that
\begin{align}
h(Z|X,E_R) &\overset{(d)}{\geq} \int_0^\infty \int_0^x p_X(x) f_{E_R}(e_R)\notag\\
&~~~\times h(Z|X=x,E_R=e_R) de_Rdx\notag\\
 &~~~+ \int_0^\infty \int_0^\infty \int_x^\infty p_X(x)f_{E_T}(e_T)f_{E_R}(e_R)\notag\\
 &~~~\times h(Z|X=x,E_R=e_R,E_T=e_T) de_Rde_Tdx\notag\\
\end{align}
where $(d)$ follows since conditioning reduces entropy.

We then have
\begin{align}
h(Z|X,E_R) &\overset{(e)}{\geq} h(N + E_T) \int_0^\infty \int_0^x f_X(x)f_{E_R}(e_R)de_Rdx\notag\\
&~~~ + g(c^*)\notag\\
&\overset{(f)}{=} h(N + E_T)\int_0^\infty f_{E_R}(x) \mathrm{Pr}(X > x)dx\notag\\
&~~~ + g(c^*),
\end{align}
where $(e)$ is obtained using Lemma~\ref{lem:e_bound} and $(f)$ follows from integration by parts.
\begin{lemma}\label{lem:e_bound}
\begin{align}
V &= \int_0^\infty\int_0^\infty \int_x^\infty h(Z|X=x,E_R=e_R,E_T=e_T)p_X(x)\notag\\
&~~~\times f_{E_R}(e_R)f_{E_T}(e_T)de_Rdxde_T\notag\\
& \geq -\int_{c^*}^\infty f_N(u)\log f_N(u)du,
\end{align}
where $c^*$ is given by (\ref{eq:cstar}).
\end{lemma}
\begin{proof}
See Appendix~\ref{app:lemma_proof}
\end{proof}

Next, we need the Paley-Zygmund inequality, which is
\begin{align}
\mathrm{Pr}(X \geq \theta \mathbb{E}[X]) \geq (1 - \theta)^2\frac{(\mathbb{E}[X])^2}{\mathbb{E}[X^2]},
\end{align}
where $0 < \theta < 1$ and $\mathbb{E}[X] = m$. Continuing our argument, we have
\begin{align}
h(Z|X,E_T) &\geq h(N + E_T)\int_0^m f_{E_R}(x) \mathrm{Pr}(X > x)dx + g(c^*)\notag\\
&\overset{(g)}{\geq} h(N + E_T) \int_0^1 f_{E_R}(\theta m) (1 - \theta)^2 \notag\\
&~~~\times \frac{m^2}{\mathbb{E}[X^2]}md\theta + g(c^*)\notag\\
&= h(N + E_T)\frac{m^2}{\mathbb{E}[X^2]} \int_0^m f_{E_R}(u)\notag\\
&~~~\times  \left(1 - \frac{u}{m}\right)^2 du + g(c^*),
\end{align}
where $(g)$ follows from the Paley-Zygmund inequality. Observe that the integral is bounded, with an upper limit of the integral of $m$. This is to ensure that the Paley-Zygmund inequality can be applied.

Finally, using the variance constraint $\mathbb{E}[X^2] \leq a$, we obtain the bound.
\end{proof}

\section{Proof of Lemma~\ref{lem:e_bound}}\label{app:lemma_proof}

\begin{proof}
The conditional pdf of $Z$ is given by
\begin{align}
f_{Z|X,E_R,E_T}(z) = \left\{
           \begin{array}{ll}
             0, & \mathrm{if}~z < 0 \\
             \mathrm{Pr}(N < -x - e_T + e_R), & \mathrm{if}~z = 0 \\
             f_N(z - x - e_T + e_R), & \mathrm{if}~z > 0.
           \end{array}
         \right.
\end{align}
We then have
\begin{align}\label{eq:entropy}
& h(Z|X = x,E_R = e_R, E_T = e_T)\notag\\
&= - \left[\mathrm{Pr}(N < -x - e_T + e_R)\log \mathrm{Pr}(N < -x - e_T + e_R)\right.\notag\\
&~~~\left. + \int_0^\infty f_N(z + e_R - x - e_T)\log f_N(z + e_R - x - e_T)dz\right].
\end{align}
There are two cases: $c = e_R - x - e_T < 0$, for which $h(Z|X = x,E_R = e_R, E_T = e_T) \geq h(N)\mathbf{1}_{h(N) < 0}$; and $c > 0$. In the latter case, we can observe that
\begin{align}
-\int_c^\infty f_N(u)\log f_N(u)du \geq -\int_{c^*}^\infty f_N(u)\log f_N(u)du,
\end{align}
where $c^*$ is given by (\ref{eq:cstar}). This follows since the inverse Gaussian pdf has a single maximum. Putting the two cases together yields the result.
\end{proof}

\bibliographystyle{ieeetr}
\bibliography{Timing_Channel}

\end{document}